\documentclass[12pt,reqno]{amsart}
\usepackage{amsthm}
\usepackage{amsmath}
\usepackage{latexsym}
\usepackage{amsfonts}
\usepackage{amssymb}
\usepackage{color}
\usepackage{bbm,dsfont}
\usepackage{graphicx}
\usepackage{hyperref}
\usepackage{subfigure}

\newtheorem{proposition}{Proposition}
\newtheorem{proposition?}{Proposition?}
\newtheorem{theorem}{Theorem}

\theoremstyle{definition}

\newtheorem{definition}{Definition}

\newcommand{\real}{\mathbb R} 
\newcommand{\complex}{\mathbb C} 
\newcommand{\nat}{\mathbb N} 
\newcommand{\integer}{\mathbb Z} 
\newcommand{\half}{\frac{1}{2}} 

\newcommand{\hi}{\mathcal{H}} 
\newcommand{\lh}{\mathcal{L(H)}} 
\newcommand{\kb}[2]{|#1\rangle\langle#2|} 
\newcommand{\tr}[1]{\mathrm{tr}\left[#1\right]} 
\newcommand{\ptr}[1]{\mathrm{tr}_1[#1]} 

\newcommand{\id}{\mathbbm{1}} 

\newcommand{\A}{\mathsf{A}}
\newcommand{\B}{\mathsf{B}}
\newcommand{\F}{\mathsf{F}}
\newcommand{\G}{\mathsf{G}}
\newcommand{\X}{\mathsf{X}}
\newcommand{\Y}{\mathsf{Y}}
\newcommand{\Z}{\mathsf{Z}}
\newcommand{\Mo}{\mathsf{M}}
\newcommand{\M}{\mathcal{M}}

\newcommand{\ebc}{\mathbf{EBC}}
\newcommand{\ibc}{\mathbf{IBC}}
\newcommand{\nibc}{\mathbf{n}-\mathbf{IBC}}
\newcommand{\mibc}{\mathbf{m}-\mathbf{IBC}}

\newcommand{\ch}{\mathbf{C}}

\newcommand{\wnc}{\Gamma^{\mathrm{wn}}} 

\newcommand{\h}{\mathcal}

\begin{document}
\title[]{Incompatibility breaking quantum channels}

\begin{abstract} 
A typical bipartite quantum protocol, such as EPR-steering, relies on two quantum features, entanglement of states and incompatibility of measurements.
Noise can delete both of these quantum features.
In this work we study the behavior of incompatibility under noisy quantum channels. The starting point for our investigation is the observation that compatible measurements cannot become incompatible by the action of any channel. 
We focus our attention to channels which completely destroy the incompatibility of various relevant sets of measurements. We call such channels \emph{incompatibility breaking}, in analogy to the concept of entanglement breaking channels. 
This notion is relevant especially for the understanding of noise-robustness of the local measurement resources for steering. 
\end{abstract}

\author[Heinosaari]{Teiko Heinosaari$^\natural$}
\address{$\natural$ Turku Centre for Quantum Physics, Department of Physics and Astronomy, University of Turku, FI-20014 Turku, Finland}
\email{teiko.heinosaari@utu.fi}

\author[Kiukas]{Jukka Kiukas$^\flat$}
\address{$\flat$ School of Mathematical Sciences, University of Nottingham, University Park,
Nottingham, NG7 2RD, UK}
\email{jukka.kiukas@nottingham.ac.uk}

\author[Reitzner]{Daniel Reitzner$^{\clubsuit}$}
\address{$\clubsuit$ Institute of Physics, Slovak Academy of Sciences, D\'ubravsk\'a cesta 9, 845 11 Bratislava, Slovakia}
\email{daniel.reitzner@savba.sk}

\author[Schultz]{Jussi Schultz$^\dagger$}
\address{$\dagger$ Dipartimento di Matematica, Politecnico di Milano, Piazza Leonardo da Vinci 32, I-20133 Milano, Italy, and Turku Centre for Quantum Physics, Department of Physics and Astronomy, University of Turku, FI-20014 Turku, Finland}
\email{jussi.schultz@gmail.com}

\maketitle

\section{Introduction}

In a typical quantum information task such as quantum key distribution or teleportation, one party, Alice, prepares a bipartite quantum system in some state, sends one of the systems to another party, Bob, after which they perform some measurements on their respective component systems. In order for such a setup to yield some advantage over a corresponding classical scenario, it is crucial that it relies on some genuine quantum feature of its constituent parts. Usually this means that the state shared by Alice and Bob should be entangled.  In order to make use of an entangled state, Alice and Bob need to perform appropriate quantum measurements, and there can be a qualitative feature that these measurements must have.
A notable requirement is {\em incompatibility}, i.e., Alice and Bob need measurements that cannot be performed jointly with a common device.
In particular, incompatibility is essential for one-sided  quantum key distribution protocols based on Einstein-Podolsky-Rosen steering \cite{BrCaWaScWi12,QuVeBr14,UoMoGu14}. 
For this reason, both entanglement and incompatibility can be viewed as quantum resources whose understanding is of utmost importance in view of emerging technological applications.

Incompatible quantum measurements, understood mathematically as POVMs without a common refinement, have been studied for a long time; see e.g. \cite{Ludwig64, BuLa84, Busch86, dMMa90} for early studies and \cite{WoPeFe09,HeMiRe14,HeKiRe15} for some recent contributions. 
Their importance was further emphasised by the recently observed \cite{QuVeBr14,UoMoGu14} tight connection to EPR-steering, which is currently under active investigation; see e.g. \cite{JoWiDo07,SkCa15,SkNaCa14,HaEbStSaFrWeSc12,PiWa15}.
We note that incompatibility should not be confused with the more restricted concept of noncommutativity, or the related concept of uncertainty principle. 
Both of these have more restricted meaning than incompatibility.

A delicate step in any quantum information protocol is the transmission of quantum systems, and their time evolution in a noisy environment. 
These processes typically induce decoherence on the quantum state, degrading its quality as a resource for some quantum protocol under consideration. 
Motivated by the steering protocols, we look at situations where the noise is local, and acts only on Alice's side. Such an effect is described by a quantum channel $\Lambda$ which in the Schr\"odinger picture maps an initial bipartite state $\varrho$ into the final state $(\Lambda_*\otimes {\rm id}) (\varrho)$. 
A particular instance is the white noise channel, which turns maximally entangled states into isotropic states.

In order to evaluate the performance of a protocol, it is essential to study the effect of noise channels on its resources. An important step in this direction is to characterize those channels which destroy a resource completely; in the case of entanglement, the relevant objects are entanglement breaking channels, i.e., channels $\Lambda$ such that $(\Lambda_*\otimes {\rm id}) (\varrho)$ is a separable state for every bipartite state $\varrho$. 
The structure of entanglement breaking channels is well known \cite{HoShRu03}, and the concept has also been generalised into various directions \cite{ChKo06,KoHoHo12}.

We now wish to change the viewpoint from entanglement breaking channels to something more appropriate for e.g. steering-based protocols for which entanglement is necessary but not sufficient. 
Owing to the duality relation
\begin{equation}\label{eqn:duality}
\tr{(\Lambda_*\otimes {\rm id}) (\varrho) A\otimes B} = \tr{\varrho  \Lambda(A)\otimes B},
\end{equation}
we can alternatively view the effect of the noise channel in the Heisenberg picture as a cause of disturbance on Alice's measurements. In other words, instead of seeing the effect of the noise as decoherence on the \emph{nonlocal resource} $\varrho$, we interpret it as a distortion on Alice's \emph{local measurement resource}, that is, incompatibility. Study of this resource can be done without any regard to the bipartite setting.

In this work we regard incompatibility as a general multi-purpose quantum resource that may be lost due to an action of a noisy quantum channel; the purpose of the paper is to initiate the study of the properties of channels relevant for this phenomenon. The starting point for our investigation is the observation that compatible measurements cannot become incompatible by the action of any channel, while the reverse happens typically. In this paper, we focus our attention to channels which completely destroy the incompatibility of various relevant sets of measurements. We call such channels \emph{incompatibility breaking}, in analogy to the concept of entanglement breaking channels. More generally, one can also consider channels that partially destroy incompatibility, as quantified e.g. by an incompatibility monotone introduced in \cite{HeKiRe15,BuHeScSt13}; this will be considered in a forthcoming paper. Furthermore, we restrict ourselves to the finite-dimensional setting, thereby excluding continuous variable Gaussian systems, which we also postpone to a separate paper.

We stress that the notion of incompatibility breaking channel is relevant for understanding noise-robustness of the restricted local measurement resources for EPR-steering. However, we also wish to make clear that the relevance of incompatibility goes beyond the steering context.

We show below that entanglement breaking channels are all incompatibility breaking in the strongest sense, i.e. they destroy incompatibility of every set of measurements. We then demonstrate examples of channels that destroy the incompatibility of any set of $n\geq 2$ observables but are not entanglement breaking. We also show that there exist channels that are incompatibility breaking in the strongest sense, but nevertheless not entanglement breaking. In this sense, entanglement is more robust to noise than incompatibility.

The paper is organized as follows.
In Sec. \ref{sec:steering} we explain the separation of local and nonlocal resources in a correlation experiment, and review the connection between incompatibility and steering. We then proceed to introduce the concepts relevant for incompatiblity breaking channels in Sec. \ref{sec:ibc}.
In Sec. \ref{sec:wn} we investigate the example of white noise mixing channels. Using the derived results we prove in Sec. \ref{sec:ebc} that every entanglement breaking channel is incompatibility breaking in the strongest sense, but the converse does not hold. Finally, the structural connections between different notions of incompatibility breaking channels are summarized in Sec. \ref{sec:summary}.

\section{Incompatibility as the local resource for steering}\label{sec:steering}

While the incompatibility of quantum devices is an interesting topic by itself, its particular role in the steering scenario has recently been investigated by many authors (see references cited in the introduction). For this reason we begin our present study with a short review of this connection, presented in a way that makes a clear separation between the nonlocal resource (bipartite entangled state) and the local resource (incompatible set of measurements). 
This separation of resources is not explicit in the notion of ``assemblage'' or ``ensemble'' of conditional states often used as a basic concept for steering (see e.g. \cite{SkNaCa14,WiJoDo07}). As a basic reference for steering we use the seminal paper \cite{WiJoDo07} by Wiseman et al.

\subsection{Incompatibility of quantum observables}\label{sec:joint}

The state dependence of measurement outcome distributions in a quantum measurement is described by the associated \emph{observable}, which is mathematically described as a \emph{positive operator valued measure} (POVM).
A POVM $\A$ with a finite outcome set $\Omega_\A$ is defined as a map $a\mapsto \A(a)$ that assigns a bounded operator to each outcome and satisfies
\begin{align*}
\sum_{a\in\Omega_\A} \A(a)=\id \quad \mathrm{and} \quad \A(a)\geq 0
\end{align*}
for all $a\in\Omega_\A$. Given a state $\varrho$ of the system, the probability to get a particular outcome $a\in \Omega_\A$ is then ${\rm tr}[\varrho \A(a)]$.
It is convenient to denote $\A(X) = \sum_{a\in X} \A(a)$ for any set $X\subseteq\Omega_\A$.
The labeling of outcomes is not relevant for the questions that we will investigate. 
For this reason, we may assume $\Omega_\A \subset \integer$ whenever this is convenient. 

A finite collection $\A_1,\ldots,\A_n$ is said to be \emph{compatible} (or jointly measurable) if there exists a \emph{joint observable} $\G$, which is an observable defined on the product outcome space $\Omega_{\A_1}\times\cdots\times\Omega_{\A_n}$ and satisfies the marginal conditions
\begin{align*}
& \G(X_1\times\Omega_{\A_2}\times\cdots\times\Omega_{\A_n}) =  \A_1(X_1)\\
& \qquad\vdots\qquad\qquad\qquad\vdots  \\
& \G(\Omega_{\A_1}\times\cdots\times\Omega_{\A_{n-1}} \times X_n) = \A_n(X_n)
\end{align*}
for all $X_1\subseteq\Omega_{\A_1},\ldots,X_n\subseteq\Omega_{\A_n}$. 
Here we have used the notation $\G(Y) = \sum_{g\in Y} \G(g)$, so e.g. the condition that the first equation is valid for all $X_1\subseteq\Omega_{\A_1}$ is equivalent to the requirement that
\begin{align*}
\sum_{a_2,\ldots,a_n} \G(a_1,a_2,\ldots,a_n) = \A_1(a_1)  
\end{align*}
holds for all $a_1\in\Omega_{\A_1}$.

This formulation of compatibility often appears in the literature; see e.g. \cite{Lahti03}. 
However, an equivalent formulation in terms of general postprocessing is often more convenient, and it also allows to formulate compatibility of an infinite number of observables. 
To properly formulate joint measurements for infinite number of observables, we first recall the general definition of a POVM.
A POVM $\G$ with infinite outcome set $\Omega$ must be understood literally as an operator-valued \emph{measure}, i.e., a map that associates a positive operator $\G(X)$ to each Borel subset $X\subset \Omega$ (an event),  has $\sum_i \G(X_i)=\G(\cup_i X_i)$ for any disjoint collection of sets ($\sigma$-additivity), and satisfies the normalisation $\G(\Omega)=\id$.

Now, in order to motivate the idea of postprocessing, suppose that we perform a measurement of an observable $\G$ and obtain an outcome $g\in \Omega$. We can make as many copies of this outcome as we want since this is just classical information.
After copying, we can process each copy of $g$ in an independent way. In particular, we can assign to it a new outcome $a$ with a conditional probability $f(a,g)$ of $a$ given $g$. 
These are normalised as $\sum_{a} f(a,g)=1$,  and we can think of the new outcomes arising from a usual finite-outcome observable $\A$, whose elements are thus defined by
\begin{equation}\label{postprocessing}
\A(a) = \int_{\Omega_{\G}} f(a,g)\G(dg), \text{ for all } a\in \Omega_\A.
\end{equation}
In general, an arbitrary collection $\mathcal M$ of finite-outcome observables is said to be \emph{compatible} (or jointly measurable) if there exists an observable $\G$ such that every $\A\in \mathcal M$ can be obtained from $\G$ by a postprocessing of the form \eqref{postprocessing}. If a collection $\mathcal M$ is not compatible, it is said to be \emph{incompatible}. 
It can be shown \cite{AlCaHeTo09} that in the case of a finite collection $\mathcal M$, this formulation of compatibility is equivalent to the usual one given above.

As an example of an infinite set of compatible observables, we recall the joint measurement of noisy spin-1/2 observables \cite{MLQT12}.
For a spin-$1/2$ quantum system, the measurement of the spin component in the direction $\vec{n}\in\real $ is described by the two-outcome observable
$$
\mathsf{S}^{\vec{n}} (\pm) = \frac{1}{2}( \id  \pm \vec{n}\cdot\vec{\sigma}).
$$
By adding uniform noise to this observable, we obtain the corresponding noisy versions
$$
\mathsf{S}^{\vec{n}}_t (\pm) =t\mathsf{S}^{\vec n} (\pm)  + \frac{1-t}{2}\id =  \frac{1}{2}( \id  \pm t\vec{n}\cdot\vec{\sigma}).
$$
The {\em spin direction observable} $\mathsf{D}$ with outcomes on the surface of the unit sphere $\mathbb{S}^2$  in $\real^3$ is defined as
$$
\mathsf{D} (X) = \frac{1}{4\pi} \int_X (\id + \vec{k}\cdot\vec{\sigma})\, {\rm d}\vec{k}.
$$ 
For any direction $\vec{n}$, we define the postprocessing function $f_{\vec{n}}$ by 
$$
f_{\vec{n}} (\pm, \vec{k})= \left\{\begin{array}{lll} 1, & \text{ if } \pm \vec{k}\cdot \vec{n} > 0, \\
0, & \text{ otherwise}. \end{array}\right.
$$
A direct calculation then shows that 
$$
\int f_{\vec{n}} (\pm, \vec{k})\, {\rm d}\mathsf{D} (\vec{k}) = \frac{1}{4\pi} \int f_{\vec{n}} (\pm, \vec{k}) (\id + \vec{k}\cdot\vec{\sigma})\, {\rm d}\vec{k} = \mathsf{S}^{\vec{n}}_{1/2} (\pm).
$$
In conclusion, the infinite set of noisy spin observabes 
$$
\{\mathsf{S}^{\vec{n}}_{1/2} : \vec{n}\in\real^3, \Vert \vec{n}\Vert =1\}
$$ 
is compatible.
It has been also shown that a corresponding set with $t>\frac{1}{2}$ is incompatible \cite{UoMoGu14}.

\subsection{Separation of local and nonlocal resources in the steering setup}\label{sec:separation}

The operational starting point for our following discussion is that of a \emph{correlation experiment} consisting of two parties, Alice and Bob, choosing observables from sets $\mathcal{M}_A$ and $\mathcal{M}_B$, respectively; see Fig. \ref{fig:ab}.
\begin{figure}
\centering
 \includegraphics[width=10cm]{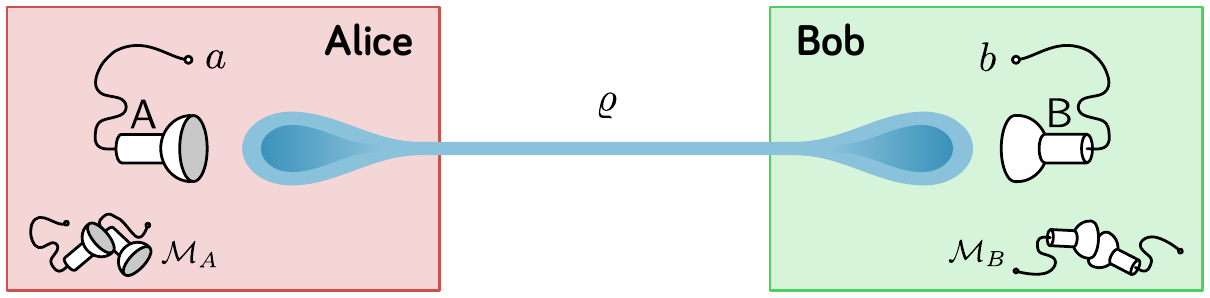}
           \caption{A correlation experiment consists of two parties, Alice and Bob, who choose observables from sets $\mathcal{M}_A$ and $\mathcal{M}_B$, respectively. The experiment contains both non-local resources (state $\varrho$) and local resources (sets $\mathcal{M}_A$ and $\mathcal{M}_B$). In the one-way steering scenario (from Alice to Bob), the relevant local resource is Alice's set $\mathcal M_A$, which  needs to be incompatible.}     \label{fig:ab}
\end{figure}
All observables $\A\in\mathcal{M}_A$ and $\B\in\mathcal{M}_B$ are assumed to have finite outcome sets $\Omega_\A$ and $\Omega_\B$. 
The experiment is described by a \emph{correlation table}
\begin{equation*}
\mathbb P(a,b|\A,\B),\quad \A\in \mathcal M_A, a\in \Omega_\A \, , \B\in \mathcal M_B \, , b\in \Omega_\B \, , 
\end{equation*}
consisting of conditional probabilities. 
The correlation table is said to have a \emph{local classical model}, if there exists a probability space $(\Omega,p)$, and response functions $f_\A(a,g)$ and $h_\B(b,g)$ such that $\sum_{a}f_\A(a,g)=\sum_b h_\B(b,g)=1$ for all $g$, and
\begin{equation}
\mathbb P(a,b|\A,\B)=\int_{\Omega} f_\A(a,g)h_\B(b,g) p(dg).
\end{equation}
If such a model does not exist, it is customary to say that the correlations are \emph{nonlocal} or \emph{nonclassical}. 

Assuming that the underlying system is quantum, i.e., $\A$ and $\B$ are POVMs and there is bipartite state $\varrho$ such that 
\begin{equation}
\mathbb P(a,b|\A,\B)={\rm tr}[\varrho(\A(a)\otimes\B(b))] \, , 
\end{equation}
 we then define for each $\A(a)$ the corresponding conditional states 
 \begin{equation}
 \sigma_{a,\A}={\rm tr}_A[\varrho(\A(a)\otimes \id)]
 \end{equation}
  on Bob's side. 
This family satisfies the \emph{no-signaling conditions}
\begin{align}
\sigma_B:=\sum_a \sigma_{a,\A}&=\sum_a \sigma_{a,\A'}, \text{ for all }\A,\A',
\end{align} 
and the normalisation ${\rm tr}[\sigma_B]=1$. 
The family $\{ \sigma_{a,\A}\}$ is often called an \emph{assemblage} or \emph{ensemble} in the context of steering.

A given assemblage is said to be \emph{non-steerable} (from Alice to Bob) if it has a local classical model, where Bob's response functions are of particular form, namely $h_\B(b,g)={\rm tr}[\varrho_g \B(b)]$ for some (measurable) family of \emph{hidden states} $\{\varrho_g\}_{g\in \Omega}$. 
Since we assume Bob to have full set of measurements, this is equivalent to the decomposition
\begin{equation}
\sigma_{a,\A}=\int_\Omega f_\A(a,g)\varrho_g \ p(dg) \, , 
\end{equation}
meaning that the assemblage can be classically simulated using the hidden states $\varrho_g$. When this is not possible, we say that the assemblage is \emph{steerable} from Alice to Bob, or that Alice can steer Bob's state.

The problem with the ``assemblage'' notion is that it mixes the local resource $\mathcal M_A$ and the shared resource $\varrho$. 
We now separate them by writing
\begin{equation}
\sigma_{a,\A}=S(\A(a)) \, , 
\end{equation}
where $S$ is the map from Alice's observable algebra to the set of subnormalised states on Bob's side, defined by 
\begin{equation}
S(A) = {\rm tr}_A[\varrho(A\otimes \id)] \, .
\end{equation}
This maps positive operators to positive operators, and satisfies the normalisation condition $S(\id)=\varrho_B$. If $S$ has a positive inverse, we can define a unique POVM $\G$ on Alice's side, for which 
\begin{equation}
S(\G(X)):=\int_X\varrho_g p(dg) \, .
\end{equation}
 Then
\begin{equation}
\A(a) =\int_\Omega f_{\A}(a,g) \G(dg) \, , 
\end{equation}
for all $\A\in \mathcal M_A$.
This means exactly that the set $\mathcal M_A$ is compatible according to the definition we gave in Subsec. \ref{sec:joint}. 
Conversely, if the measurements $\mathcal M_A$ are compatible, one can always decompose a joint POVM $\G$ into hidden states, showing that the assemblage $\{S(\A(a))\}$ is steerable. 
A typical case where $S$ has a positive inverse is when $\varrho$ is pure and has full Schmidt rank; since the topic of the present paper is Alice's local resource, we do not study the properties of $S$ in more detail here.

It is important to note that if $S$ does not have a positive inverse, the above connection between incompatibility and steering does not hold. A typical way of ending up with such a situation is to have local noise, as discussed in the introduction: suppose we begin with an $S_0$ coming from a state with full Schmidt rank, and apply a noise channel $\Lambda$ on Alice's side. If we consider the noise acting in the Schr\"odinger picture on the ``nonlocal'' state resource, we adjust the map $S$:
$$
S(A)=S_0\circ \Lambda \, .
$$
However, since the noise is local, it is more appropriate to consider it as acting on the local resource and hence in the Heisenberg picture.
This gives the following ``noisy version'' of the connection between incompatibility and steering: 

\begin{proposition} Suppose Alice and Bob share a bipartite pure state $\varrho_0$ of full Schmidt rank. 
\begin{itemize}
\item[(a)] (Ideal scenario) Alice can steer Bob with a set of measurements $\mathcal M_A$, if and only if $\mathcal M_A$ is incompatible.
\item[(b)] (Noisy scenario). Suppose that Alice has an incompatible set $\mathcal M_A$ of measurements, but subjected to a local noise channel $\Lambda$ on her system. Then Alice can steer Bob if and only if $\Lambda$ does not map $\mathcal M_A$ into a compatible set of observables.
\end{itemize}
\end{proposition}

In conclusion, noisy local resources for EPR-steering essentially have two requirements, incompatible set $\mathcal M_A$ of measurements, and noise channels that do not \emph{break the incompatibility} of that set. The rest of the paper is focused on properties of such channels.

\section{Incompatibility breaking channels}\label{sec:ibc}

\subsection{$n$-incompatibility breaking channels}

Quantum evolution is generally described by a \emph{quantum channel}, which is a unital completely positive linear map $\Lambda $ on the observable algebra $\lh$ (the set of bounded operators on the system Hilbert space $\mathcal H$). Quantum channels describe operations that can be performed on the system either actively (by e.g. unitary quantum gates) or passively (by environmental interaction).  We denote by $\ch$ the set of all quantum channels on $\lh$  and use the notation $\Lambda_1 \circ \Lambda_2$ for the functional composition of two channels $\Lambda_1,\Lambda_2\in\ch$, i.e., 
\begin{align*}
(\Lambda_1 \circ \Lambda_2)(T) = \Lambda_1( \Lambda_2(T)) \, .
\end{align*} 
This is called \emph{concatenation} of $\Lambda_2$ and $\Lambda_1$.

A quantum channel $\Lambda$ maps  any observable $\A$ into another observable $\Lambda(\A)$ by way of composition:
\begin{align*}
\Lambda(\A)&:=\Lambda\circ \A \, , & (\Lambda\circ \A)(x)&:=\Lambda(\A(x)) \, .
\end{align*}
We immediately notice that complete positivity is not necessary for this to be a valid transformation of observables, but the same is true for any unital positive linear map. The meaning of the following simple observation is that such maps cannot create incompatibility \cite{HeKiRe15}. (Implications of this result to steering have been noticed earlier; see \cite{QuEtal15} and the references therein.)

\begin{proposition}\label{prop:unital_positive}
Let $\Lambda$ be a unital and positive linear map on $\lh$.
If the observables $\A_1,\ldots,\A_n$  are compatible, then also the transformed observables $\Lambda(\A_1), \ldots, \Lambda(\A_n)$ are compatible.
\end{proposition}

\begin{proof}
Assuming $\{\A_1,\ldots,\A_n\}$ is compatible, there exists a joint observable $\G$ for $\A_1,\ldots,\A_n$.  
From the definition it immediately follows that $\Lambda(\G)$ is a joint observable  for  $\Lambda(\A_1),\ldots,\Lambda(\A_n)$.
\end{proof}

A unitary channel $\sigma_U(T):=UTU^*$ is reversible, and it follows that observables $\A_1,\ldots,\A_n$ are compatible if and only if $\sigma_U(\A_1),\ldots,\sigma_U(\A_n)$ are compatible.
In other words, unitary channels preserve incompatibility. 
As an extreme example of the opposite kind, the completely depolarising channel $T \mapsto \tr{\eta T}\id$ maps every operator to a multiple of the identity operator.
Hence, the image of any collection of observables is compatible.
A typical channel is something between these two extremes, and we thus expect that it maps some subsets of observables into compatible ones, but not all. 

\begin{definition}
Let $\Lambda$ be a channel and $\mathcal A$ a subset of observables. If the image $\Lambda(\h A)$ is compatible, we say that $\Lambda$ {\em breaks the incompatibility of $\h A$}. 
\end{definition}

As an example, let us consider quantum channels $\Gamma_{t,\Theta,\eta}:\lh\to\lh$ of the form
\begin{align}
\Gamma_{t,\Theta,\eta} (T) = t \Theta(T) + (1-t) \tr{\eta T}\id \, , 
\end{align}
where  $\eta$ is a fixed  state, $\Theta$ is a fixed channel, and $0\leq t \leq 1$. The channel $\Gamma_{t,\Theta,\eta}$ is thus a mixture of $\Theta$ and the completely depolarising channel $T \mapsto \tr{\eta T} \id$, and one can  regard $\Gamma_{t,\Theta,\eta}$ as a noisy version of $\Theta$.
Since the completely depolarising channel breaks the incompatibility of any collection of observables, we expect that the same is true for any $\Gamma_{t,\Theta,\eta}$ with $t\leq t_0$ for some critical value $t_0$.
In the following we give a bound that is independent of $\Theta$.

\begin{proposition} \label{prop:nIBC-noisy}
Let $n\geq 2$ be an integer.  
For all $0\leq t \leq \frac{1}{n}$, the quantum channel $\Gamma_{t,\Theta,\eta}$ breaks the incompatibility of an arbitrary set $\{\A_1,\ldots, \A_n\}$ of $n$ observables.
\end{proposition}
\begin{proof}
For a collection of $n$ observables $\A_1,\ldots,\A_n$, we define a map $\G$  by
\begin{align*}
\G(a_1,\ldots, a_n) & = t \Theta(\A_1(a_1)) \cdot\Pi_{j\neq 1} \tr{\eta \A_j(a_j)}  \\
& + \cdots  \\
& + t \Theta(\A_n(a_n)) \cdot\Pi_{j\neq n} \tr{\eta \A_j(a_j)} \\
& + (1-tn) \cdot \Pi_{j} \tr{\eta \A_j(a_j)} \, .
\end{align*}
Then $\G(a_1,\ldots, a_n)\geq 0$ and
\begin{align*}
\G(X_1\times\Omega_{\A_2}\times\ldots\times\Omega_{\A_n}) = \Gamma_{t,\Theta,\eta} (\A_1(X_1)) 
\end{align*}
and similarly for the other marginals. 
Hence, $\G$ is a joint observable for the observables $\Gamma_{t,\Theta,\eta}(\A_1), \ldots, \Gamma_{t,\Theta,\eta}(\A_n)$.
\end{proof}

Motivated by this example we now introduce the following concept. 

\begin{definition} 
Let $\Lambda$ be a quantum channel on $\lh$. If $\Lambda$ breaks the incompatibility of every collection of $n$ observables, then $\Lambda$ is called \emph{$n$-incompatibility breaking}. We let $\nibc$ denote the set of all $n$-incompatibility breaking channels.
\end{definition}

The following basic properties of the sets $\nibc$ are immediate.
\begin{proposition}\label{prop:basic_nibc}
Each $\nibc$ is a convex subset of $\ch$. With respect to the channel concatenation relation, $\nibc$ is an ideal in $\ch$, that is, if $\Lambda\in\ch$ and $\Lambda'\in\nibc$, then $\Lambda \circ \Lambda',\,\Lambda' \circ \Lambda\in\nibc$. The following inclusions hold:
\begin{equation}\label{basic_inclusions}
\nibc \subseteq \cdots \subseteq 3-\ibc \subseteq 2-\ibc \, .
\end{equation}
\end{proposition}
\begin{proof} In order to show convexity, we take $\Lambda,\Lambda'\in\nibc$ and $0\leq t \leq 1$, and $n$ observables $\A_1,\ldots,\A_n$.
Assuming that $\Lambda,\Lambda'\in\nibc$, the sets $\{ \Lambda(\A_1),\ldots,\Lambda(\A_n)\}$ and $\{ \Lambda'(\A_1),\ldots,\Lambda'(\A_n)\}$ have joint observables $\G$ and $\G'$, respectively.
The mixture $t\G + (1-t)\G'$ is then a joint observable for the observables $(t\Lambda + (1-t)\Lambda')(\A_1),\ldots,(t\Lambda + (1-t)\Lambda')(\A_n)$. Hence $t\Lambda + (1-t)\Lambda' \in\nibc$, i.e. $\nibc$ is convex. The fact that $\nibc$ is an ideal follows immediately from the definition, together with Prop. \ref{prop:unital_positive}. The inclusions are obvious.
\end{proof}
We will show in Sec. \ref{sec:different} that every subset $\nibc$ is strictly containing some higher subset $\mibc$, at least if the dimension of the Hilbert space is large enough.

\subsection{Incompatibility breaking channels and complete incompatibility}

We now proceed to introduce the key concept of the paper in its general form.

\begin{definition}
Let $\Lambda$ be a quantum channel on $\lh$. If $\Lambda$ breaks the incompatibility of the set of all observables, then $\Lambda$ is called \emph{incompatibility breaking}. We denote by $\ibc$ the set of all such channels.
\end{definition}

Since the above definition requires that $\Lambda$ maps all observables into a compatible set, the task of determining if a given channel is incompatibility breaking appears tedious. In view of this, it would be desirable to have some smaller ``test sets'' of observables. This motivates the next definition.
\begin{definition} 
A set $\mathcal A$ of observables is said to have \emph{complete incompatibility}, if any channel that breaks its incompatibility, is necessarily in $\ibc$.
\end{definition}

We remark that any set having complete incompatibility is an optimal resource for noisy EPR-steering, assuming that the form of the noise is unknown. 
In order to demonstrate basic examples of sets having complete incompatibility, we make the following observation.
\begin{proposition}\label{prop:convex}
If $\Lambda$ breaks the incompatibility of a set of observables $\mathcal A$, then it also breaks the incompatibility of the convex hull of the set of all postprocessings of elements of $\mathcal A$ (into other finite-outcome POVMs).
\end{proposition}
\begin{proof} The claim concerning postprocessing follows immediately from the transitivity of the postprocessing relation, i.e. if $\mathcal M$ is compatible with $\A(a) =\int f_\A(a,g) \G(dg)$ for each $\A\in \mathcal M$, and we postprocess each $\A$ further into $\B(b) =\sum_a h(b,a) \A(a)$ using an arbitrary postprocessing function $h$, the resulting $\B(b)=\int f_\B(b,g) \G(dg)$ is a postprocessing of $\G$ with $f_\B(b,g)=\sum_a h(b,a)f_\A(a,g)$. Furthermore, given two of such observables with the same outcome set, the convex combination $\B=\lambda \B_1+(1-\lambda)\B_2$ is also a postprocessing of $\G$ with function $f(b,g)=\lambda f_{\B_1}(b,g)+(1-\lambda)f_{\B_2}(b,g)$. Hence, the total set of observables obtained from $\G$ in this way is compatible.
\end{proof}

An observable $\A$ is called rank-1 if each nonzero operator $\A(x)$ has rank 1.
Using Prop.~\ref{prop:convex} we get the following result, exhibiting one basic set having complete incompatibility.

\begin{proposition}\label{rank1tib} 
The set of all rank-1 observables with at most $d^2$ outcomes has complete incompatibility, where $d$ is the dimension of the Hilbert space.
\end{proposition}

\begin{proof} We can represent each finite-outcome observable as a convex combination of extremal POVMs. 
All extremals have at most $d^2$ outcomes, and every extremal is a relabeling of some rank one extremal \cite{HaHePe12}. 
Since relabelling is a special case of postprocessing, Proposition \ref{prop:convex} completes the proof. 
\end{proof}

With this result, we are able to prove a connection between incompatibility breaking channels, and the previously defined concept of $n$-incompatibility breaking channels. 

\begin{proposition}\label{prop:cap}
$\ibc =\bigcap_{n\geq 2} \nibc$.
\end{proposition}

\begin{proof}
The inclusion $\ibc\subset \bigcap_{n\geq 2} \nibc$ is trivial. For the converse, let $\h A$ denote the set of all observables with at most $d^2$ outcomes. By Proposition \ref{rank1tib}, $\h A$ has complete incompatibility, so in order to establish $\bigcap_{n\geq 2} \nibc\subset \ibc$ it suffices to prove that $\Lambda (\h A)$ is compatible for $\Lambda\in\bigcap_{n\geq 2}\nibc$. 
By adding trivial outcomes if necessary, we can assume that each $\A\in\h A$ has the outcome set $\Omega = \{1,2,\ldots, d^2\}$. 
By Tychonov's theorem, the cartesian product set $\Omega_0 = \Omega^{\h A}$ is compact in the product topology. 
For each $\A\in\h A$, let $\pi_\A:\Omega_0\to\Omega$ be the canonical projection. If $\Lambda \in \bigcap_{n\geq 2}\nibc$, then for any finite collection $\h F\subset\h A$ the image $\Lambda(\h F)$ is compatible, and there exists a joint observable $\tilde \G$ on the finite product space $\Omega^{\h F}$. Since $\Omega_0=\Omega^{\h F}\times \Omega^{\mathcal A\setminus \mathcal F}$, we can trivially extend $\tilde \G$ to an observable $\G_{\mathcal F}$ on $\Omega_0$ by defining $\G_{\mathcal F}(X\times Y)=\tilde \G(X)\mu(Y)$ for all $X\subset \Omega^{\h F}$ and $Y$ in the Borel $\sigma$-algebra $\mathcal{B}(\Omega^{\mathcal A\setminus \mathcal F})$ of subsets of the (possibly infinite) set $\Omega^{\mathcal A\setminus \mathcal F}$, where $\mu$ is any probability measure. Then $\A = \G_{\h F}\circ \pi_\A^{-1}$ for all $\A\in\h F$.

Now for any $\A\in\h A$, let $\h M_\A$ denote the set of all observables $\G$ on $\Omega_0$ such that $\Lambda(\A) = \G \circ \pi_\A^{-1}$. Then the above results imply that the collection $\{ \h M_\A\}_{\A\in\h A}$ has the finite intersection property. Hence, if we can make the sets $\h M_\A$ compact by putting a suitable topology on the set of all POVMs, we can conclude that $\bigcap_{\A\in\h A} \h M_\A\neq \emptyset$, showing that $\Lambda (\h A)$ is compatible.

In order to establish suitable compactness, we let $C(\Omega_0)$ denote the Banach space of continuous complex valued functions on $\Omega_0$, and denote by $\h L(C(\Omega_0),\lh)$ the Banach space of bounded linear maps $C(\Omega_0)\to\lh$. By the POVM analogue of the Riesz-Markov-Kakutani representation theorem \cite{MMMPI80}, the regular Borel POVMs\footnote{A regular Borel POVM on $\Omega_0$ is a positive operator valued measure $\G:\mathcal B(\Omega_0)\to \lh$ such that the total variation of each complex measure $X\mapsto {\rm tr}[T\G(X)]$ is a regular Borel measure for each (trace class) operator $T$.} on $\Omega_0$ are in one-to-one correspondence with positive normalized ($\M(1) =\id$) elements of $\h L(C(\Omega_0),\lh)$. We equip $\h L(C(\Omega_0),\lh)$ with the locally convex Hausdorff vector topology coming from the separating family of seminorms $p_{(T,f)} (\Mo) = \vert\tr{T \Mo(f)}\vert$ with $T\in\lh$ (recall that $\hi$ is finite dimensional) and $f\in C(\Omega_0)$. Using standard arguments, it is easy to see that the closed unit ball $\mathcal M_1$ of $\mathcal L(C(\Omega_0),\lh)$ is compact in this topology \footnote{In fact, for each pair $(T,f)\in \lh\times C(\Omega_0)$, define the compact sets $S_{T,f} = \{ \lambda \in \complex \mid |\lambda|\leq \|T\|_1\|f\| \}$ (where $\|T\|_1$ is the trace norm). By Tychonov's theorem, the cartesian product of these topological spaces is also compact. The product is the set of functions $B:\lh\times C(\Omega_0)\to \complex$ with $|B(T,f)|\leq \|T\|_1\|f\|$, equipped with the topology of pointwise convergence. By the standard duality, the subset of bilinear functions $B$ correspond exactly to the elements of $\mathcal L(C(\Omega_0),\lh)$, and this subset is closed because the limit of a pointwise convergent net of bilinear maps is clearly bilinear. Hence $\mathcal M_1$ is compact.}. The subset
\begin{equation}
\bigcap_{T\geq 0,f\geq 0} \{ \Mo\in \mathcal M_1\mid \tr{T\Mo(f)}\in [0,\infty)\}\subset \mathcal M_1
\end{equation}
coincides with the set of all observables. It is clearly closed and therefore compact as well.

Then each $\h M_\A$ is closed in the above topology: if $(\Mo_i)_{i\in\h I}$ is a net of elements of $\h M_\A$ converging to a map $\Mo\in \h L(C(\Omega_0),\lh)$ and $g:\Omega\to\complex$ is a function, then $g\circ\pi_\A\in C(\Omega_0)$ and hence
$$
\tr{T \Mo (g\circ \pi_A)} = \lim_i \tr{T\Mo_i (g\circ \pi_\A)} = \lim_i \tr{T\A (g)} = \tr{T\A(g)}
$$
for all $T\in\lh$. In other words, $\Mo \circ \pi_\A^{-1} = \A$. Hence $\h M_\A$ is closed and therefore compact. The proof is complete.
\end{proof}

The above result thus gives another way of checking whether or not a channel is incompatibility breaking or not, namely, by checking the compatibility of the images $\Lambda(\A_1),\ldots, \Lambda(\A_n)$ of finite sets of observables. 

\section{Breaking incompatibility with white noise}\label{sec:wn}

As a special instance of the previously defined class of noisy channels $\Gamma_{t,\Theta,\eta}$, we have the class of \emph{white noise mixing channels} $\wnc_t \equiv \Gamma_{t,{\rm id},\id/d}$, $0\leq t < 1$. The action of these types of channels is very simple:
\begin{align*}
(\wnc_t)_*(\varrho)&=t\varrho+(1-t)\frac{\id}{d} & \wnc_t(A)&= t A +(1-t) \frac1d {\rm tr}[A]\id.
\end{align*}
Here the parameter $1-t$ represents the amount of white noise (understood as the completely depolarising channel) mixed into the state $\varrho$. 
With $t=1$ the channel is identity, hence clearly does not break the incompatibility of any set. The values of $1-t$ for which $\wnc_t$ becomes $n$-incompatibility breaking and incompatibility breaking, respectively, represent the robustness of the incompatibility of the corresponding sets of observables against white noise. 
Similar ideas on noise-robustness of incompatibility have been recently considered in \cite{HeKiRe15,Ha15} for pairs of observables.

\subsection{Robustness of incompatibility of finite collections of observables}

For white noise mixing channels we have the following improvement of Prop. \ref{prop:nIBC-noisy}. It describes the amount of white noise needed to turn all finite collections of observables compatible.

\begin{proposition} \label{prop:nIBC-wn}
The channel $\wnc_t$ is $n$-incompatiblity breaking for all
\begin{equation}
\label{eq:clonebound}
0\leq t\leq\frac{n+d}{n(d+1)} \, 
\end{equation}
where $d=\dim \hi$.
\end{proposition}

\begin{proof}
Following the approximate cloning scheme of \cite{KeWe99}, we can proceed in the same way as in \cite{HeScToZi14} to find a sufficient condition for the compatibility of noisy versions of any $n$ observables.  
Let $\{\A_1,\ldots,\A_n\}$ be a set of observables. 
We define $\G$ as
\[
\G(x_1,x_2,\ldots,x_n)=\frac{d}{{d+n-1 \choose n}}\ptr{S_n \A_1(x_1)\otimes \A_2(x_2)\otimes\cdots\otimes \A_n(x_n)S_n}
\]
for all $x_j\in\Omega_j$, where $S_n$ is the projection from $\hi_d^{\otimes n}$ to its symmetric subspace and $\ptr{\cdot}$ is partial trace over all but one part of the system.
A direct calculation shows that $\G$ is a joint observable for the observables $\wnc_t(\A_1), \ldots,\wnc_t(\A_n)$ with $t=\frac{n+d}{n(d+1)}$.
\end{proof}

\subsection{$2-\ibc\neq 3-\ibc$}\label{sec:different}

As a consequence of Prop. \ref{prop:nIBC-wn} we can demonstrate that there are channels that are $2$-incompatibility breaking but not $3$-incompatibility breaking. 
Let us consider three qubit observables 
\begin{align}
\X(\pm 1) = \half ( \id \pm \sigma_x) \, , \label{eq:sigmax}\\
\Y(\pm 1) = \half ( \id \pm \sigma_y) \, , \label{eq:sigmay}\\
\Z(\pm 1) = \half ( \id \pm \sigma_z) \, . \label{eq:sigmaz}
\end{align}
The image of $\X$ under the action of $\wnc_t$ is
\begin{align*}
\wnc_t(\X)(\pm 1) = \half ( \id \pm t\sigma_x) \, , 
\end{align*}
and similarly for $\Y$ and $\Z$.
Therefore, if we choose $t$ such that $1/\sqrt{3}<t$, then the observables $\wnc_t(\X)$, $\wnc_t(\Y)$ and $\wnc_t(\Z)$ are incompatible \cite{} and hence $\wnc_t$ is not $3$-incompatibility breaking.
But if $t$ also satisfies $t\leq 2/3$, then by Prop. \ref{prop:nIBC-wn} ($d=2$, $n=3$) the channel $\wnc_t$ is $2$-incompatibility breaking.
In conclusion, $3-\ibc\subsetneq 2-\ibc$.

In the following we want to extend the previous observation and to show that every set $\nibc$ is strictly containing some higher set $\mibc$ at least in some Hilbert space with high enough dimension. 
For this purpose, we recall that for every integer $n\geq 3$, it is possible to construct a set of $n$ incompatible observables such that every subset of $n-1$ observables is compatible \cite{KuHeFr14}.
We say that this kind of set of observables is a \emph{Specker set of order $n$} (see  \cite{LiSpWi11} for an explanation of Specker's parable of the overprotective seer).

The explicit construction presented in \cite{KuHeFr14} uses a Clifford algebra $CL(m)$ of $m$ generators for a Specker set of order $m$,  and the observables are very similar to those in \eqref{eq:sigmax}--\eqref{eq:sigmaz}.
The dimension of the Hilbert space is then the same as the chosen representation of  $CL(m)$.
If $m$ is even, then $CL(m)$ has a single irreducible representation of degree $2^{m/2}$, and if $m$ is odd, then $CL(m)$ has two irreducible representations, each of degree $2^{(m-1)/2}$. Furthermore, for odd $m$ we can make use of the explicit form of one of the representations that is known.

\begin{proposition} \label{prop:nIBC-different}
For every integer $n\geq 3$ and any odd integer $m\geq n^2$, the sets $\nibc$ and $\mibc$ are different in the Hilbert space of dimension $2^\frac{m-1}{2}$.
\end{proposition}

\begin{proof}
Fix an integer $p \geq (n^2-1)/2$ and a matrix representation of Clifford algebra in the Hilbert space of dimension $d=2^p$. 
The representation consists of $m=2p+1$ selfadjoint matrices $\delta_j$ that define $m$ different binary observables $\A_j$ as
\[
\A_j (\pm):=\half(\id \pm  \delta_j) \, .
\]
We have
\begin{align*}
\wnc_t(\A_j)(\pm) = \half(\id \pm  t \delta_j) \, , 
\end{align*}
hence from \cite{KuHeFr14} we know that the observables $\wnc_t(\A_1),\ldots,\wnc_t(\A_m)$ are incompatible if and only if $t>\frac{1}{\sqrt{m}}$.
That is, the channel $\wnc_t$ is not $m$-incompatibility breaking for any $t>\frac{1}{\sqrt{m}}$.
On the other hand, if $t$ satisfies the inequality in \eqref{eq:clonebound}, then $\wnc_t$ is $n$-incompatibility breaking.
There exists a $t$ that simultaneously satisfies these two inequalities if
\begin{align}
\frac{1}{\sqrt{2p+1}}<\frac{n+2^p}{n(2^p+1)} \, ,
\end{align}
which is equivalent to
\begin{align}
n < \frac{2^p \sqrt{2p+1}}{2^p+1-\sqrt{2p+1}}  \, .
\end{align}
This inequality is satisfied since $n\leq\sqrt{2p+1}$.
\end{proof}

\subsection{Robustness of incompatibility of the set of all measurements}

We now proceed to determine the noise parameter $t$ for which the white noise channel becomes incompatibility breaking. This can conveniently be done using the hidden state models appearing in the context of steering. In fact, Wiseman et al.\ \cite{JoWiDo07} use Werner's construction \cite{Werner89} to obtain hidden state models for the isotropic states to derive steerability conditions for the set of projective measurements. A generalisation by Almeida et al.\ \cite{AlPiBaToAc07} (based on \cite{Ba02}) involves rank-1 POVMs; it has been adapted to steering context in \cite{QuEtal15}. In these models, the probability space for the hidden variable is the unitary group $U(d)$, with the normalised Haar measure ${\rm d}U$, and the hidden states are simply given by
$$
U\mapsto U|\varphi_0\rangle\langle \varphi_0|U^*,
$$
where $\varphi_0$ is some fixed fiducial state the choice of which does not play a role.
Consequently, the joint observable associated with the incompatibility breaking property of the white noise channel is the $U(d)$-covariant normalised POVM
\begin{align}
\G(Z)=d\int_Z U|\varphi_0\rangle\langle\varphi_0|U^* \,{\rm d}U.
\end{align}
Using the steering connection together with Prop. \ref{rank1tib}, we can translate the results of \cite{AlPiBaToAc07} into statements of incompatibility breaking with the white noise channel. 
In order to keep our paper reasonably self-contained, we provide sketches of the proofs, adapted to our context of incompatible measurements. 
The first result, adopted from \cite{JoWiDo07}, concerns the noise-robustness of the total set of all projective measurements $\mathcal{P}$.

\begin{proposition}\label{prop:projections}
The white noise channel 
$\wnc_t$ breaks the incompatibility of $\mathcal P$ if 
\begin{align}
t \leq t_\mathcal{P}\equiv\frac{1}{d-1}(-1+\sum_{k=1}^d\frac 1k) \, .
\end{align}
\end{proposition}

\begin{proof}[Sketch of proof]
By Prop.~\ref{prop:convex} we can restrict to nondegenerate observables, i.e., observables given by $\A(i) = \vert \psi_i\rangle\langle\psi_i\vert$ where $\{\psi_i\}$ is an orthonormal basis of $\hi$. For any such $\A$, define the postprocessing function
$$
f_\A(i,U)=\begin{cases} 1, & \text{if } \langle U\varphi_0|\A(i) U\varphi_0\rangle=\max_{j} \langle U\varphi_0|\A(j) U\varphi_0\rangle\\
0 & \text{ otherwise }\end{cases}
$$
so that 
$$
\int f_\A(i,U)\, {\rm d}\G(U) = d\int f_\A(i,U) U  \vert \varphi_0\rangle\langle\varphi_0 \vert U^* \, {\rm d}U.
$$
The trick is now to express the vector $U\varphi_0$ as 
\begin{equation}\label{eqn:trick}
U\varphi_0 = \sum_{k=1}^d (x_k + iy_k) \psi_k = \sum_{k=1}^d \sqrt{u_k} e^{i\theta_k} \psi_k
\end{equation}
where $u_k= x_k^2 + y_k^2$ so the integral over the unitary group is replaced with the integral over the unit sphere of $\complex^d$, that is, 
\begin{equation}\label{eqn:measure}
\mathrm{d}{U} = N_d \, \delta \left(1-\sum_{l=1}^d u_l\right)\prod_{k=1}^d\mathrm{d}u_k \mathrm{d}\theta_k
\end{equation}
where $\delta$ is the delta function and $N_d$ is the normalization factor. Since now $\langle U\varphi_0|\A(j) U\varphi_0\rangle = u_j$, the Markov kernel is simply
$$
f_\A(i,{\bf u}) = \prod_{j=1}^d\Theta (u_i-u_j)
$$
where $\Theta$ is the Heaviside step function. It is now a straightforward calculation to see that 
$$
\int f_\A(i,U)\, {\rm d}\G(U)  = t_\mathcal{P} \A(i) + (1-t_\mathcal{P}) \frac{1}{d} \id = \wnc_{t_\mathcal{P}} (\A(i))
$$
which completes the proof.
\end{proof}

We denote by $\mathcal{R}_1$ the set of all rank-1 observables with finite number of outcomes. 
Analogously to Prop. \ref{prop:projections}, we can obtain from \cite{AlPiBaToAc07} a sufficient condition when $\wnc_t$ breaks the incompatibility of $\mathcal{R}_1$. 
This leads to the following conclusion, noted also in \cite{QuEtal15,Pusey15}.

\begin{proposition} 
The white noise mixing channel $\wnc_t$ breaks the incompatibility of $\mathcal{R}_1$ if
\begin{align}\label{eq:rid-ibc}
t \leq t_0 \equiv \frac{(3d-1)(d-1)^{d-1}}{(d+1)d^d} \, .
\end{align}
\end{proposition}

\begin{proof}[Sketch of proof]
Following the method of \cite{Ba02}, for an observable $\A\in\mathcal{R}_1$ we define the postprocessing
\begin{align*}
f_\A (i,U) &= \Theta \bigg(\langle U\varphi_0 \vert\A(i) U\varphi_0\rangle - \tr{\A(i)}/d\bigg) \langle U\varphi_0 \vert \A(i) U\varphi_0\rangle \\
& \quad + \frac{\tr{\A(i)}}{d}\sum_j \langle U\varphi_0 \vert \A(j) U\varphi_0\rangle \\
&\quad\times \left( 1-   \Theta \bigg(\langle U\varphi_0 \vert\A(j) U\varphi_0\rangle - \tr{\A(i)}/d\bigg) \right)\\
\end{align*}
where $\Theta$ is again the Heaviside function. Let $n$ denote the number of outcomes of $\A$. Since $\A$ is a rank-1 POVM, there exist unit vectors $\phi_1,\ldots, \phi_n\in \hi$ and numbers $\alpha_1,\ldots, \alpha_n\in (0,1]$ such that $\A(i) = \alpha_i \vert \phi_i\rangle\langle \phi_i\vert $. 
This implies that 
$$
\Theta \bigg(\langle U\varphi_0 \vert\A(j) U\varphi_0\rangle - \tr{\A(j)}/d\bigg) = \Theta \left(\vert\langle U\varphi_0 \vert\phi_j \rangle\vert^2 - 1/d\right) 
$$ 
and we can proceed as in the proof of Prop.~\ref{prop:projections} by fixing an orthonormal basis of $\hi$, expressing $U\varphi_0$ in this basis, and switching the integration over the unitary group to integration over the unit sphere of $\complex^d$. After a lengthy calculation, we then obtain 
$$
\int f_\A (i,U) U\vert \varphi_0\rangle\langle \varphi_0 \vert U^*\, \mathrm{d}U = t_0\A(i) + (1-t_0) \frac{\tr{\A(i)}}{d} \id = \wnc_{t_0} (\A(i)).
$$
\end{proof}

According to Prop. \ref{rank1tib}, the set $\mathcal{R}_1$ has complete incompatibility, and hence is large enough to determine whether a given channel is incompatibility breaking.
We thus conclude the following.

\begin{theorem}
The white noise mixing channel $\wnc_t$ is incompatibility breaking if\eqref{eq:rid-ibc} holds.
\end{theorem}

\begin{figure}
\centering
            \includegraphics[width=10cm]{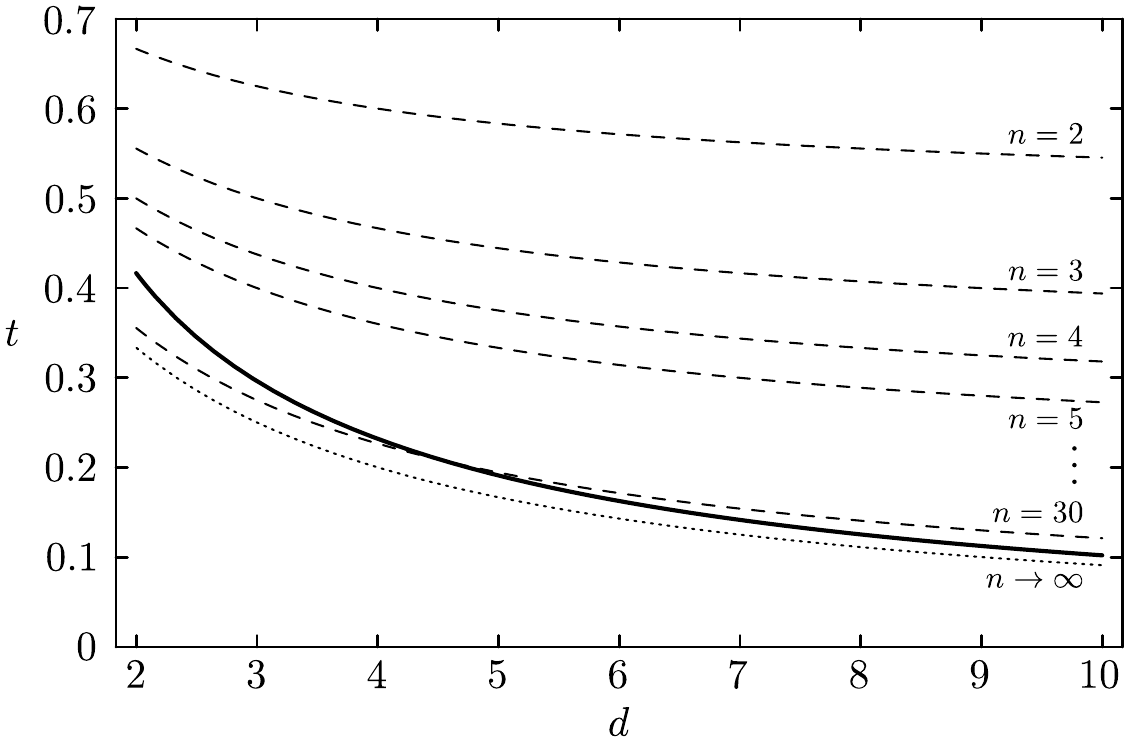} 
           \caption{The dashed lines correspond to the bound given in \eqref{eq:clonebound} for different $n$ and the solid line corresponds to the bound given in \eqref{eq:rid-ibc}.}
    \label{fig:bounds}
\end{figure}

We note that for large $n$ and small $d$, the bound in \eqref{eq:rid-ibc} can be smaller than the one given in \eqref{eq:clonebound}; see Fig. \ref{fig:bounds} for comparison.

\section{Connection to entanglement breaking channels}\label{sec:ebc}

In this section we make an important connection to the existing literature on decoherence-inducing channels: we will prove that the \emph{entanglement breaking} channels form a proper subclass of incompatibility breaking channels. A similar result has been recently obtained in \cite{Pusey15}.

\begin{theorem}\label{thm:main_theorem}
Every entanglement breaking channel is incompatibility breaking, but the converse does not hold. 
\end{theorem}

We recall that a quantum channel $\Lambda$ is called \emph{entanglement breaking} if the bipartite state $(\Lambda_*\otimes {\rm Id})(\varrho)$ is separable for all initial states $\varrho$. In particular, this means that the local classical model always exists for any collection of measurements performed after the application of the channel. We denote by $\ebc$ the set of all entanglement breaking channels. The structure of these channels is well known: in a finite dimensional Hilbert space every entanglement breaking channel $\Lambda$ can be written in the form 
\begin{align}\label{eq:ebc}
\Lambda(T)=\sum_x \tr{\varrho_xT} \F(x) \, , 
\end{align}
where $\F$ is an observable with a finite number of outcomes and each $\varrho_x$ is a state \cite{HoShRu03}.
 
It is easy to see that every entanglement breaking channel is $n$-incompatibility breaking for all $n$, i.e., $\ebc\subseteq\nibc$.
Namely, let $\Lambda$ be as in \eqref{eq:ebc} and suppose that $ \A_1,\ldots,\A_n$ are incompatible observables.
We define an observable $\G$ by
\begin{align*}
\G(a_1,\ldots, a_n) = \sum_x \tr{\varrho_x \A_1(a_1)}\cdots \tr{\varrho_x \A_n(a_n)} \F(x) \, .
\end{align*}
Then
\begin{align*}
\G(X_1\times\Omega_{\A_2}\times\ldots\times \Omega_{\A_n}) = \sum_x \tr{\varrho_x \A_1(X_1)} \F(x)
\end{align*}
for $X_1\subseteq\Omega_1$, and similarly for the other marginals. 
Hence, $\G$ is a joint observable for the observables $\Lambda(\A_1), \ldots, \Lambda(\A_n)$.
This together with Prop. \ref{prop:cap} implies that every entanglement breaking channel is incompatibility breaking.
One can also see this directly by noticing that an observable $\A$ is mapped into 
\begin{equation}
\Lambda(\A(a)) = \sum_x \tr{\varrho_x \A(a)} \F(x), 
\end{equation}
that is, each $\Lambda(\A)$ is  postprocessing of the observable $\F$. 
This means that all observables $\Lambda(\A)$ are compatible, hence, $\Lambda$ is incompatibility breaking.

Our next observation is that there are also other $n$-incompatibility breaking channels than just the entanglement breaking channels.
To see this, we fix an orthonormal basis $\{\varphi_j\}_{j=1}^d$ of $\hi$ and denote 
\begin{align*}
\psi_0 = \frac{1}{\sqrt{d}} \sum_{j=1}^d \varphi_j \otimes \varphi_j \, .
\end{align*}
The pure state $\kb{\psi_0}{\psi_0}$ is maximally entangled, and a mixed state 
\begin{align*}
t \kb{\psi_0}{\psi_0} + (1-t) \frac{1}{d^2} \id
\end{align*}
 is entangled if and only if $t > \frac{1}{1+d}$ \cite{HoHo99}.
We recall that a channel $\Gamma$ is entanglement breaking if and only if the state $(\Gamma_*\otimes {\rm id})(\kb{\psi_0}{\psi_0})$ is separable \cite{HoShRu03}. 
It follows that a white noise channel $\wnc_t$ is entanglement breaking if and only if $t \leq \frac{1}{1+d}$.
A comparison with Prop. \ref{prop:nIBC-wn} shows that for each $2\leq n\leq  \dim\hi$, we have $\ebc\subsetneq\nibc$.

The remaining question is: are there other incompatibility breaking channels than just the entanglement breaking channels?
To see that the answer is positive, we consider again the white noise mixing channel $\Gamma^{\rm wn}_t$, which is entanglement breaking if and only if  $t\leq 1/(d+1)$.
The upper bound $1/(d+1)$ is smaller that the upper bound $t_0$ in \eqref{eq:rid-ibc}, hence choosing any $t$ between $1/(d+1) < t \leq t_0$ gives a channel $\wnc_t$ that is incompatibility breaking but not entanglement breaking.

What remains is to show that such a $t$ can be chosen, i.e., that $t_0>1/(d+1)$. To see this, we first write $t_0$ as 
\begin{align}
t_0 = \frac{1}{d+1}\left( 3-\frac{1}{d}\right) \left( 1-\frac{1}{d}\right)^{d-1}.
\end{align}
We can then see that $(1-1/d)^{d-1}$ is a monotonically decreasing function of $d$ with its limit being $1/e$; hence we find
\begin{align}
t_0\geq\frac{1}{d+1}\left(\frac{3}{e}-\frac{1}{de}\right).
\end{align}
The term in brackets is strictly larger than one for $d\geq 4$. For $d=2,3$ it is a simple exercise to check that $t_0 > 1/(d+1)$. Moreover, with increasing $d$ the gap between $t_0$ and $1/(d+1)$ closes, as we can bound $t_0$ easily also from above by $3/(d+1)$.
In conclusion, $\ebc \subsetneq \ibc$.

\section{Summary}\label{sec:summary}

\begin{figure}
\centering
            \includegraphics[width=9cm]{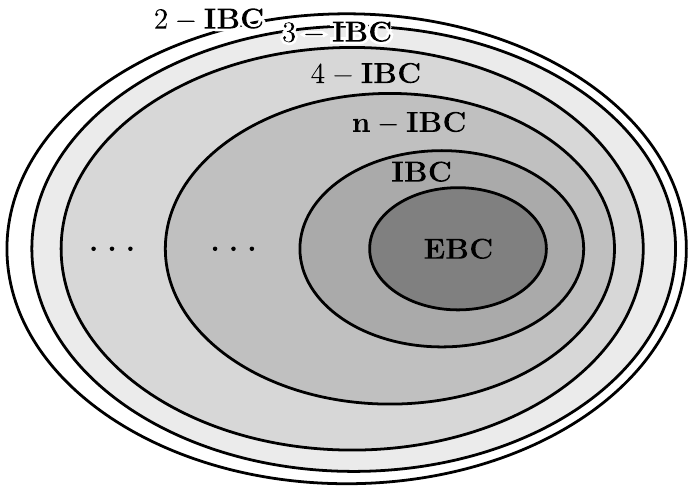} 
           \caption{Depiction of the inclusions of sets $\nibc$, $\ibc$ and $\ebc$. While we know that $\ebc\subsetneq\ibc$, $3-\ibc\subsetneq 2-\ibc$ and $\mibc\subsetneq\nibc$ for every $n$ and odd $m\geq n^2$, the strict inclusion of other combinations of sets is unclear. Still, this shows that there is a chain of strict inclusions of e.g.~sets $3^j-\ibc$ for $j\in\nat$.}
    \label{fig:sets}
\end{figure}

We have initiated the study of the evolution of quantum incompatibility under noisy channels, by considering the case where the channel completely destroys the incompatibility of relevant sets of measurements. In particular, we have defined the set $\nibc$ of channels that break the incompatibility of each collection of $n$ measurements. These sets are included within each other forming a chain $n-\ibc\subseteq\ldots\subseteq 3-\ibc\subseteq 2-\ibc$ (see Fig.~\ref{fig:sets}). While the strict inclusion of all these sets is unknown, the inclusion $3-\ibc\subsetneq 2-\ibc$ is strict as well as is the whole chain of inclusions of sets $3^j-\ibc$ for all $j\in\nat$.

Concerning the set of channels that are $\nibc$ for all $n$ we form an important set $\ibc$ of channels breaking the incompatibility of the total set of all observables. 
Furthermore, we have made an interesting connection between $\ibc$ and the set of all entanglement breaking channels, $\ebc$. 
While it is easy to show that $\ebc\subseteq\ibc$, we have also used rather nontrivial hidden variable models from \cite{Ba02} to show that there are examples of incompatibility breaking channels which do not break entanglement.

Concerning further research perspectives on this topic, one could look at channels that reduce incompatibility according to some relevant measure \cite{HeKiRe15}, without destroying it completely. Another related direction is to investigate the Heisenberg evolution of incompatible measurements in an open quantum system.

\section*{Acknowledgments}

The authors wish to thank Matthew Pusey for bringing his results from \cite{Pusey15} to our attention.
DR acknowledges support from Slovak Research and Development Agency grant APVV-0808-12 QIMABOS, VEGA grant QWIN 2/0151/15 and programme SASPRO QWIN 0055/01/01.
JS acknowledges support from the Italian Ministry of Education, University and Research (FIRB project RBFR10COAQ). JK acknowledges support from the EPSRC project EP/J009776/1.

\end{document}